\documentclass[conference]{IEEEtran}

\IEEEoverridecommandlockouts                              

\usepackage{amsmath}
\usepackage{parskip}
\usepackage{listings}
\usepackage{indentfirst}
\usepackage{graphicx}
\usepackage{float}
\usepackage{extarrows}
\usepackage{amssymb}
\usepackage{amsthm}
\usepackage{subfigure}
\usepackage{epstopdf}
\usepackage{caption}
\usepackage{cite,color}
\usepackage{todonotes}
\usepackage{tikz,pgfplots}
\usepackage{xcolor}
\usetikzlibrary{automata,positioning}
\usetikzlibrary{arrows,shapes,chains}
\usepgflibrary{patterns}

\begin{document}
\title{\LARGE \bf   Benefits of Coding on Age of Information in  Broadcast Networks}
 \author{Xingran Chen and Shirin Saeedi Bidokhti\\
	\small 
	 Department of Electrical and Systems Engineering\\
University of Pennsylvania\\ \{xingranc, saeedi\}@seas.upenn.edu\\	
}

\IEEEoverridecommandlockouts
\newtheorem{lemma}{Lemma}
\newtheorem{note}{Note}
\newtheorem{property}{Property}
\newtheorem{theorem}{Theorem}
\newtheorem{definition}{Definition}
\newtheorem{corollary}{Corollary}
\newtheorem{proposition}{Proposition}
\newtheorem{remark}{Remark}
\newtheorem{assumption}{Assumption}

\maketitle
\thispagestyle{empty}
\pagestyle{empty}

\section*{Abstract}
 Age of Information (AoI) is studied in two-user broadcast networks with feedback, and lower and upper bounds are derived on the expected weighted sum AoI of the users. In particular, a class of simple coding actions is considered and within this class, randomized and deterministic policies are devised. Explicit conditions are found for symmetric dependent channels under which coded randomized policies strictly outperform the corresponding uncoded policies. Similar behavior is numerically shown for  deterministic policies.

\begin{IEEEkeywords}
Age of Information, Network Coding, Randomized Policy,
Max-Weight Policy, Feedback, Broadcast Packet Erasure Channels
\end{IEEEkeywords}

\section{Introduction}\label{sec: Introduction}

Sending status updates in a timely manner has significant importance in the Internet of Things (IoT) applications. In practice, it is not always  effective to  update the information as fast as possible for it may cause further delay in the network queues. To measure the timeliness of  information at a remote system, the concept of Age of Information (AoI) was introduced  in \cite{Tutorial-Kaul, 2011-S.Kaul, 2012-S.Kaul}. AoI measures, at the receiving side, how much time has passed since the generation time of the latest received packet. In \cite{2012-S.Kaul}, a single source and server setup were considered under First-Come First-Serve (FCFS) queue management and it was shown that there is an optimal update rate that minimizes time-average AoI. Further extensions to networks of multiple sources and servers with and without packet management were studied in \cite{2012-S.K.Kaul, 2014-C.Kam, 2016-C.Kam, 2016-M.Costa, 2019-R.D.Yates}. More recently, AoI has been studied as a performance metric  in various contexts such as  source and channel coding \cite{2018-J.Zhong, 2018-P.Mayekar, 2018-R.Devassy}, caching \cite{ 2017-C.Kam, 2017-R.D.Yates, 2018-S.Zhang}, energy harvesting \cite{2018-A.Arafa, 2018-A.Arafa1, 2018-S.Feng}, sampling \cite{2018-J.Yun, 2018-1806.03487, 2019-TasmeenZaman} and scheduling \cite{2018-Q.He, 2018-C.Joo, 2015-B.Li, 2016-B.Li, 2018-I.Kadota, 2018-IgorKadota, 2018-R.Talak, 2018-R.Talak1, 2018-R.Talak2}.

In coding theory, previous work has mainly studied point to point channels. For example, \cite{2016-K.Chen,2017-R.D.Yates1,2017-P.Parag} consider erasure channels, propose coding schemes, and analyze the resulting average or peak AoI in various setups. More recently, \cite{2019-E.Najm} proves that when the source alphabet and  channel input alphabet have the same size, a  Last-Come First-Serve  (LCFS) with no buffer policy is optimal.   Considering erasure channels with FCFS M/G/1 queues, \cite{2018-H.Sac}  finds an optimal  block length for channel coding to minimize the average age and average peak age of information. 

AoI has also been investigated in network management and scheduling. In particular, \cite{2018-Q.He} proposes scheduling policies to optimize the overall age in a wireless network. Maintaining equally up-to-date and synchronized information  from multiple sources is studied in \cite{2018-C.Joo}. In \cite{2018-IgorKadota}, scheduling algorithms are designed to minimize AoI in wireless broadcast channels. \cite{2018-I.Kadota} devises scheduling policies to minimize average AoI under throughput constraints in wireless multi-access networks. The minimum age of time-varying wireless channels with interference constraints is obtained in \cite{2018-R.Talak,2018-R.Talak1,2018-R.Talak2} with and without channel state information.

In this work, we aim to shed light on the interplay between AoI and (channel/network) coding in the context of broadcast packet erasure channels (BPECs) with feedback. BPECs and their variants have been investigated in previous work such as \cite{2009-L.Georgiadis,2010-C.Wang, 2012-M.Gatzianas, 2013-M.Gatzianas, 2018-M.Heindlmaier} and rate-optimal coding algorithms are designed using  (network) coding ideas. The key idea is due to \cite{2009-L.Georgiadis} where it is shown that the entire capacity region of two-user BPECs with feedback can be attained by XOR-ing overheard packets. More precisely, suppose a packet $p$, intended for user~$1$, is broadcasted and only  received at user~$2$. Scheduling algorithms re-transmit this packet because it is not received at its intended receiver. However, one may be able to exploit  coding opportunities by tracking such packets (and this is possible through the available feedback). For example, in a similar manner, a packet $q$, intended for user $2$, may  get transmitted and received only at receiver $1$. Now instead of re-transmitting $p$ and $q$ in two uses of the network, one can transmit the XOR packet $p\oplus q$ which is simultaneously useful for both users.

 As opposed to the aforementioned literature, in this work, we seek {\it efficiency in terms of age as opposed to rate}. The underlying challenge is as follows. On the one hand,  the highest rate of communication in  BPECs can be attained when coding is employed across packets of different users \cite{2009-L.Georgiadis}. A higher rate effectively corresponds to a smaller delay (both in the sense that the queues get emptied faster and in the sense that fewer uses of the network are needed in total to transmit a fixed number of information bits). On the other hand, to achieve high rates with coding, we have to incur delay by waiting for the arrival of other packets for the purpose of coding as well as prioritizing their transmission. So it is not clear apriori when  coding is beneficial. We will devise scheduling policies that {\it schedule different coding actions}, as opposed to traditional schemes that schedule different users, and show the benefit of coding in terms of average age over uncoded schemes such as those proposed in \cite{2018-IgorKadota}.

Motivated by the capacity achieving coding scheme of \cite{2009-L.Georgiadis}, in this work we restrict attention to a class of coding algorithms consisting of three actions: uncoded transmission for user $1$, uncoded transmission for user $2$, and coded (XOR-ed) transmission for both users. We consider a discrete time model as in \cite{2018-IgorKadota} and study the expected weighted sum of AoI (EWSAoI) at the users. The first contribution of the paper is a general lower bound on the achievable EWSAoI. As opposed to previous lower bounds  (e.g. \cite{2018-IgorKadota}) that hold only in the class of traditional scheduling algorithms, the new lower bound is valid for any coding scheme. The second contribution of the paper is the devise and analysis of EWSAoI or an upper bound on it for (i) stationary randomized policies and (ii) deterministic Max-Weight (MW) policies. In the class of randomized policies, for symmetric channels, we  find  conditions under which coded policies perform strictly better than their corresponding uncoded policies.  For MW policies,  we numerically compare the performance of coded and uncoded MW policies and show  gains of coding.

The remainder of this paper is organized as follows. The problem setup and notation are introduced in Section~\ref{sec: Basic model}. Section~\ref{sec: lower bound} presents a general lower bound on EWSAoI. In  Section~\ref{sec: stationary randomized policy}, we devise a randomized policy based on three coding actions and find a closed-form expression for the resulting EWSAoI. We further study the special case of symmetric BPECs and find conditions under which coded randomized policies strictly outperform uncoded randomized policies in terms of age. In Section~\ref{sec: Standard Max-Weight Policies}, we propose a Max-Weight policy and derive an upper bound on the resulting EWSAoI. Simulation results and the comparison of uncoded vs. coded schemes are presented in Section~\ref{sec: Numerical results and discussions}.

\section{System Model}\label{sec: Basic model}

We consider a model where time is slotted. At the beginning of every time slot, new packets are generated for the users and they replace any undelivered packets from previous time slots. Our model is similar to \cite{2018-I.Kadota, 2018-IgorKadota}, where users are scheduled  with the goal of minimizing age.

Transmission occurs on a noisy network which we model by a broadcast packet erasure channel with two users. In each time slot $k$, the input $X(k)$ to the channel is a packet. The packet is successfully delivered to user $i$ with probability $1-\epsilon_{i}$, $0\leq\epsilon_{i}<1$, and lost with probability $\epsilon_i$. Let $Z_i(k)$ be a random variable modeling erasure at user $i\in\{1,2\}$ in time slot $k\in\{1,2,\ldots\}$. We assume that the channel is memoryless and hence $\{Z_i(k)\}_{k=1}^\infty$ is an iid Bernoulli process with probability $1-\epsilon_i$. The output of the channel at user $i$ in  slot $k$ is:
\begin{align*}
Y_i(k)=\left\{\begin{array}{cc}X(k) &\text{if }Z_i(k)=1\\\Delta&\text{otherwise}\end{array}\right.
\end{align*}
where $\Delta$ is a symbol denoting erasure. Note that the pairs $\{(Z_1(k),Z_2(k))\}_k$ are independent across time (over $k=1,2,\ldots$) but potentially correlated across $(Z_1, Z_2)$. In addition, the feedback is available at the encoder after each transmission. Define $\epsilon_1,\epsilon_2,\epsilon_{12}$ as
\begin{align*}
\epsilon_{1}:=&\Pr(Z_1=0)\\
\epsilon_2:=&\Pr(Z_2=0)\\
\epsilon_{12}:=&\Pr(Z_1=0,Z_2=0),
\end{align*}
and hence we have
\begin{align*}
\Pr(Z_1=1,Z_2=0)=&\epsilon_{2}-\epsilon_{12}\\
\Pr(Z_1=0,Z_2=1)=&\epsilon_{1}-\epsilon_{12}.
\end{align*}
The statistics of $(Z_1,Z_2)$ that describes the channel is assumed fixed and given and can be characterized by $(\epsilon_1,\epsilon_2,\epsilon_{12})$.

We consider a simple class of coding algorithms that consists of three actions, including a network coding action.
The encoder is modeled by a network of virtual queues. Let  $Q_{1}^{(i)}$ denote the queue of incoming packets for user $i$ and $Q_2^{(i)}$ denote the queue of packets that are intended for user $i$ but are received only by the other user. The encoder can track such packets using the available feedback. For $i=\{1,2\}$, we use the notation $\backslash i$  as short for $\{1,2\}\backslash i$. The packets in $Q_2^{(i)}$ are not received at their intended receivers, but are received at the other receiver and act as side information for it -- this can be exploited in the code design at the encoder. In particular, the encoder can XOR packets in $Q_2^{(1)}$ with $Q_2^{(2)}$ and form more efficient coded packets for transmission.

In each time slot $k$, the encoder decides between the following three actions, denoted by $A(k)\in\{1,2,3\}$ and defined below:\vspace{-.15cm}
\begin{itemize}
  \item  $A(k)\!=\!1$: a packet is transmitted from $Q_{1}^{(1)}$;
  \item  $A(k)\!=\!2$: a packet is transmitted from $Q_{1}^{(2)}$;
  \item  $A(k)\!=\!3$: a coded packet is transmitted from $Q_{2}^{(1)}$, $Q_{2}^{(2)}$.
\end{itemize}

\begin{figure}[t!]
\centering
\includegraphics[width=2.5in]{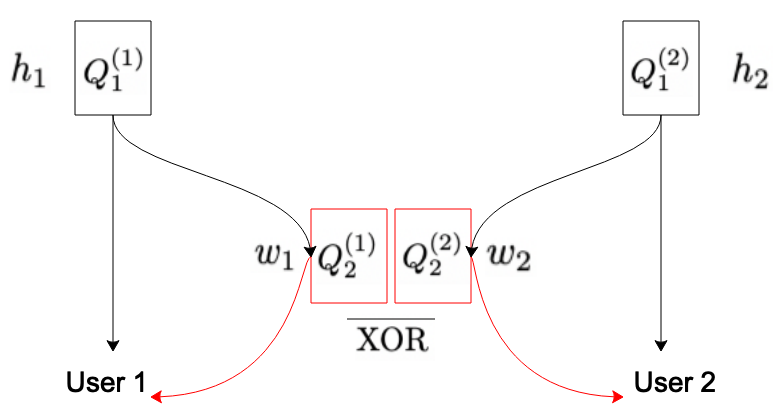}
\caption{the virtual network of queues at the encoder}
\label{fig_channels1}
\end{figure}

\begin{definition}[Age of Information \cite{Tutorial-Kaul}]\label{def: AoI}
Consider a source-destination pair. Let $\{t_k\}_k$ be the times at which packets are generated and $\{t_{k}'\}_k$ be  the times at which packets are received at the destination. At any time $\xi$, denote $N(\xi)=\max\{k|t_{k}'\leq\xi\}$, and $u(\xi)=t_{N(\xi)}$. The Age of Information (AoI) at the destination is $\Delta(t)=t-u(t)$.
\end{definition}

Using Definition \ref{def: AoI}, let $h_{i}$ be the positive real number that represents the age at user $i$. The age $h_{i}$ increases linearly in time when there is no delivery of packets to user $i$ and drops with every delivery to a value that represents how old the received packet is. 
\begin{lemma}\label{lem: buffer size 1}
To attain the optimal age in the above class of $3-$action coding algorithms,  we can assume, without loss of generality, that all queues are of buffer size $1$. 
\end{lemma}
\begin{proof}
The proof is presented in Appendix~\ref{App: Lemma buffer size 1}.
\end{proof}
 
To capture the evolution of $h_i$  in the class of $3-$action algorithms described, we proceed as follows.
First, define $w_{i}(k)$ as the (current) age of information at $Q_{2}^{(i)}$ in slot $k$. 
If the packet in $Q_{2}^{(i)}$ is successfully delivered at user $i$ by time $k$, then  it is removed from $Q_{2}^{(i)}$ and $w_{i}(k)$ is defined to be zero; if a packet in $Q^{(i)}_{2}$ is replaced by a new packet $p$ in slot $k$, then $w_{i}(t)$ is the age of the new packet  when $t>k$. More precisely, suppose packet $p$ is generated at time $t_p$. At time $t$, while $p$ is in the queue  $Q_{2}^{(i)}$,  the age $w(t)$ is $t-t_p$. Finally, once a packet is delivered successfully at user $i$ from $Q_1^{(i)}$, then the existing packet in $Q_2^{(i)}$ (which is necessarily older) becomes obsolete and hence we remove it from $Q_{2}^{(i)}$ and define  $w_i(k+1)$ to be $0$. Thus,
the recursion of $w_{i}(k)$ is
\begin{align}\label{equ: BM-w}
\footnotesize
w_{i}(k+1)=\left\{\begin{aligned}
&0\,\,\text{if}\, A(k)\in\{i,3\},\, Z_i(k)=1\\
&1\,\,\text{if}\, A(k)=i,\, (Z_i(k),Z_{\backslash{i}}(k))=(0,1)\\
&(w_{i}(k)+1)\cdot1_{\{w_{i}(k)>0\}}\,\,\text{otherwise}
\end{aligned}\right..
\end{align}
Based on $w_{i}(k)$, the age function $h_{i}(k)$ evolves as follows: 
\begin{align}\label{equ: BM-age}
\footnotesize
h_{i}(k+1)=\left\{\begin{aligned}
&1\quad\text{if}\,\, A(k)=i,\, Z_i(k)=1\\
&w_{i}(k)+1\quad\text{if}\,\, A(k)=3,\, Z_{i}(k)=1\\
&h_{i}(k)+1\quad\text{otherwise}
\end{aligned}\right..
\end{align}
\begin{remark}\label{w<=h-1}
Using the mathematical recursions in \eqref{equ: BM-w} and \eqref{equ: BM-age}, we conclude that $w_{i}(k)\leq h_{i}(k)-1$.
\end{remark}

\begin{figure}[t!]
\centering
\includegraphics[width=3.5in]{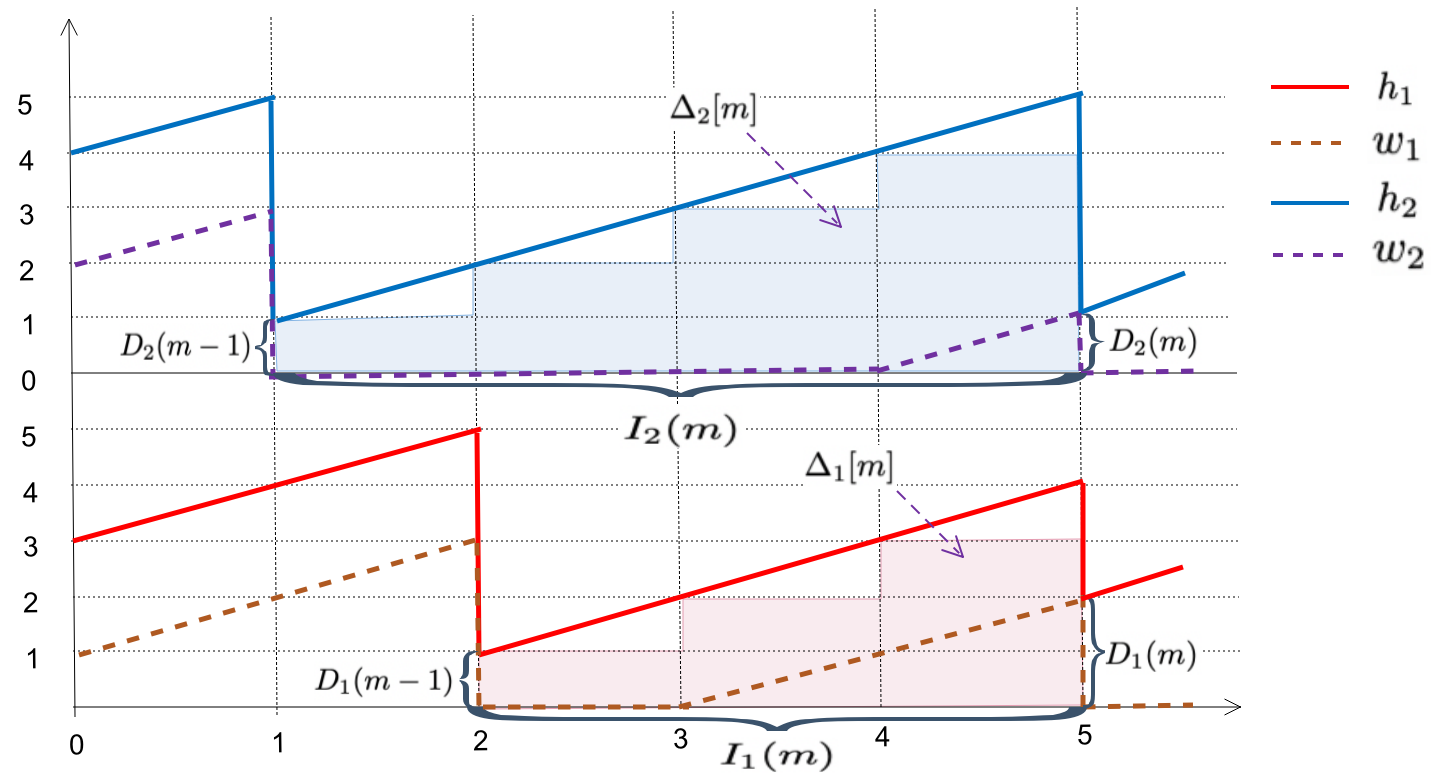}
\caption{a sample path of the channel state $(h_{1}, h_{2}, w_{1}, w_{2})$ which initial state $(h_{1},h_{2},w_{1},w_{2})=(3,4,1,2)$}
\label{fig_channels2}
\end{figure}

\subsection{A Sample  Path}
\label{sec:samplepath}
A sample path for the evolution of $w_1,h_1$ and $w_2,h_2$ is shown in Fig. \ref{fig_channels2}. The initial state is $(h_{1},h_{2},w_{1},w_{2})=(3,4,1,2)$, and the actions and the channels are as follows:

\begin{center}
\begin{tabular}{r|ccccc}
$k$&1&2&3&4&5\\
$A(k)$&2&1&1&2&3\\
$(Z_1(k),Z_2(k))$&(1,1)&(1,0)&(0,1)&(1,0)&(1,1)\\
\end{tabular}
\end{center}

 Now consider a general sample path associated  with a transmission policy and a finite time-horizon $T$. For this sample path, let $N_{i}(T)$ be the total number of packets delivered to user $i$ up to and including time slot $T$, and $I_{i}(m)$ be the number of time slots between the $(m-1)$th and $m$th deliveries to user $i$, i.e., the inter delivery times of user $i$. Denote the age of  user $i$ after delivery of the $m$th packet by $D_{i}(m)$ and let $L_{i}$ be the number of  remaining time slots after the last packet delivery to the same user. With this notation, the time-horizon can be written as 
$T=\sum_{m=1}^{N_{i}(T)}I_{i}(m)+L_{i}$ with $i\in\left\{1,2\right\}$.

Next,  consider the sum of the instantaneous ages in the interval corresponding to $I_{i}(m)$, $m\geq2$, denoted by $\Delta_{i}[m]$. As shown in Figure \ref{fig_channels2}, $\Delta_{i}[m]$ is equal to the  area  underneath the age curve in the corresponding interval minus $I_{i}(m)$ small triangle areas; i.e.,
\begin{align}
\footnotesize
&\Delta_i[m]=\sum_{k\text{ in between delivery of ${m-1}^{\text{}th}$ and $m^{\text{th}}$ packets}} h_i[k]\\
=&\frac{(D_i(m-1)+I_{i}(m))^{2}}{2}-\frac{(D_i(m-1))^{2}}{2} -\frac{I_i(m)}{2}\\
=&\frac{I_{i}^{2}(m)}{2}+D_i(m-1)I_{i}(m)-\frac{I_i(m)}{2}\label{equ: delta_i[m]}.
\end{align}

\subsection{The Expected Weighted Sum AoI}

Both $w_i(k)$ and $h_i(k)$ and their evolution depend on the policy that we choose  and, hence, we sometimes write them as  $w_{i}^{\pi}(k)$ and $h_{i}^\pi(k)$. Moreover, we denote the vector $(h_1(1),h_2(1),w_1(1),w_2(1))$ by $\vec{s}(1)$. 

We aim to find policies $\pi$ that minimize the following EWSAoI at the users:
\begin{align}\label{equ: BM-expected weighted sum aoi}
\mathbb{E}\left[\frac{1}{2T}\sum_{k=1}^{T}\sum_{i=1}^{2}\alpha_{i}h^{\pi}_{i}(k)\Big|\vec{s}(1) \right]
\end{align}
where $\alpha_{1}$ and $\alpha_{2}$ are weights associated to users $1$ and $2$, respectively. We assume $\alpha_{i}\geq0$ and $\alpha_{1}+\alpha_{2}=1$.
For notational simplicity, we omit $\vec{s}(1)$ hereafter, and hence the minimum age is given by the following optimization problem.
\begin{align}\label{equ: BM-optimization}
\min_{\pi\in\Pi}\mathbb{E}[J_{T}^{\pi}],\quad\text{where}\quad J_{T}^{\pi}=\frac{1}{2T}\sum_{k=1}^{T}\sum_{i=1}^{2}\alpha_{i}h_{i}^{\pi}(k).
\end{align}
Since we aim to minimize the EWSAoI in the long run, we define 
\begin{align*}
J^{\pi}=\lim_{T\rightarrow\infty}J_{T}^{\pi}.
\end{align*}

\section{A Lower Bound}\label{sec: lower bound}
We prove a  lower bound on EWSAoI as stated below. 
 \begin{theorem}\label{thm: lowerbound}
For any  communication policy $\pi$, we have: 
\begin{align}
\mathbb{E}[J^{\pi}]\geq\frac{1}{4}\left(\frac{(\sum_{i=1}^{2}\sqrt{\alpha_{i}(2-\epsilon_{12}-\epsilon_{\backslash{i}})})^{2}}{(1-\epsilon_{12})(2-\epsilon_{1}-\epsilon_{2})}+1\right).\end{align}
\end{theorem}
\begin{remark}
The lower bound of Theorem \ref{thm: lowerbound} holds in general and is not restricted to the class of $3-$action coded algorithms that we introduced in Section \ref{sec: Basic model}.
\end{remark}
\begin{proof} 
Consider a sample path associated  with a transmission policy and a finite time-horizon $T$ (see Section \ref{sec:samplepath}). The EWSAoI as defined in \eqref{equ: BM-expected weighted sum aoi} can be re-written in terms of $\Delta_i(m)$'s:
\begin{align}\label{equ: LB-age1}
J_{T}^{\pi}=&\frac{1}{2T}\sum_{k=1}^{T}\sum_{i=1}^{2}\alpha_{i}h_{i}(k)\nonumber\\
=&\frac{1}{2}\!\sum_{i=1}^{2}\!\frac{\alpha_{i}}{T}\!\!\left(\!\sum_{m=1}^{N_{i}(T)}\!\!\!\!\Delta_i(m)\!+\!\frac{1}{2}L_i^2\!+\!D_{i}(\hspace{-.05cm}N_{i}(\hspace{-.05cm}T\hspace{-.05cm})\hspace{-.05cm})L_i\!-\!\frac{1}{2}L_i\!\!\right)\!\!.
\end{align}
Since $D_i(m)\geq 1$ for all $1\leq m\leq N_{i}(T)$, we can lower bound \eqref{equ: delta_i[m]} and hence \eqref{equ: LB-age1} by substituting $D_i(m)=1$.
Using similar steps as \cite[Eqns. (9) - (14)]{2018-IgorKadota}, we find
\begin{equation}\label{equ: age-4}
\begin{aligned}
J^{\pi}_{T}
\geq\frac{1}{4}\sum_{i=1}^{2}\alpha_{i}\frac{T}{N_{i}(T)}+\frac{1}{4}.
\end{aligned}
\end{equation}

We remark that so far, the lower bound on $J_{T}^{\pi}$ is the same as \cite[Eqn. (7)]{2018-IgorKadota}. We now depart from \cite{2018-IgorKadota} by allowing for the general class of coding and scheduling schemes.
Recall that $N_i(T)$ is the total number of packets received by user $i$, $i=1,2$. In the limit of $T\to\infty$, $\frac{N_i(T)}{T}$ is the throughput of user~$i$. We further know  the capacity of two-user broadcast packet erasure channels from \cite{2009-L.Georgiadis}. In particular, any non-negative rate pair $(R_1,R_2)$  is achievable if and only if it satisfies the following conditions.
\begin{align}
&\frac{R_1}{1-\epsilon_1}+\frac{R_2}{1-\epsilon_{12}}\leq 1\label{equ: LB1-throughput1}\\
&\frac{R_1}{1-\epsilon_{12}}+\frac{R_2}{1-\epsilon_{2}}\leq 1\label{equ: LB1-throughput2}.
\end{align}
Hence, from \eqref{equ: LB1-throughput1} and \eqref{equ: LB1-throughput2}, we have
\begin{align}
&\lim_{T\rightarrow\infty}\frac{\frac{N_{1}(T)}{T}}{1-\epsilon_1}+\frac{\frac{N_{2}(T)}{T}}{1-\epsilon_{12}}\leq 1\label{equ: LB-throughput1}\\
&\lim_{T\rightarrow\infty}\frac{\frac{N_{1}(T)}{T}}{1-\epsilon_{12}}+\frac{\frac{N_{2}(T)}{T}}{1-\epsilon_{2}}\leq 1\label{equ: LB-throughput2}.
\end{align}
For notational simplicity, denote $x=\lim_{T\rightarrow\infty}\frac{N_{1}(T)}{T}$ and $y=\lim_{T\rightarrow\infty}\frac{N_{2}(T)}{T}$. Substituting $x$ and $y$ into (\ref{equ: LB-throughput1}) and (\ref{equ: LB-throughput2}), we obtain
\begin{align*}
&\frac{x}{1-\epsilon_{1}}+\frac{y}{1-\epsilon_{12}}\leq1\\
&\frac{x}{1-\epsilon_{12}}+\frac{y}{1-\epsilon_{2}}\leq1.
\end{align*}
Re-arranging terms, we obtain
\begin{align*}
&(1-\epsilon_{12})x+(1-\epsilon_{1})y\leq(1-\epsilon_{12})(1-\epsilon_{1})\\
&(1-\epsilon_{2})x+(1-\epsilon_{12})y\leq(1-\epsilon_{12})(1-\epsilon_{2})\\
\end{align*}
and summing the above two inequalities, we get
\begin{align}
&k_1x+k_2y\leq k_3\label{eq:ineq3}
\end{align}
where 
\begin{align*}
k_1=&2-\epsilon_{2}-\epsilon_{12}\\
k_2=&2-\epsilon_{1}-\epsilon_{12}\\
k_3=&(1-\epsilon_{12})(2-\epsilon_{1}-\epsilon_{2}).
\end{align*}

We can now use \eqref{eq:ineq3} along with the Cauchy-Schwarz inequality and write
\begin{align*}
\Big(\frac{\alpha_{1}}{x}+\frac{\alpha_{2}}{y}\Big)k_3&\geq \Big(\frac{\alpha_{1}}{x}+\frac{\alpha_{2}}{y}\Big)(k_{1}x+k_{2}y)\\
&\geq\left(\sqrt{\alpha_{1}k_1}+\sqrt{\alpha_{2}k_2}\right)^{2}.
\end{align*}
Re-arranging terms and replacing for $k_1,k_2,k_3$, we find
\begin{align*}
\Big(\frac{\alpha_{1}}{x}+\frac{\alpha_{2}}{y}\Big)\geq \frac{(\sum_{i=1}^{2}\sqrt{\alpha_{i}(2-\epsilon_{\backslash{i}}-\epsilon_{12})})^{2}}{(1-\epsilon_{12})(2-\epsilon_{1}-\epsilon_{2})}
\end{align*}
and hence, from \eqref{equ: age-4}, we obtain
\begin{align*}
J^{\pi}\geq&\lim_{T\rightarrow\infty}\frac{1}{4}\sum_{i=1}^{2}\alpha_{i}\frac{T}{N_{i}(T)}+\frac{1}{4}\\
\geq&\frac{1}{4}\Big(\frac{(\sum_{i=1}^{2}\sqrt{\alpha_{i}(2-\epsilon_{12}-\epsilon_{\backslash{i}})})^{2}}{(1-\epsilon_{12})(2-\epsilon_{1}-\epsilon_{2})}+1\Big).
\end{align*}
\end{proof}

\section{Coded Randomized Policies}\label{sec: stationary randomized policy}

Consider a stationary randomized policy where each action is chosen with a fixed probability in each time slot, independent of the system's status. Denote by $\mu_i$ the probability of action $A(k)=i$, $i\in\{1,2,3\}$, $k\in\{1,\ldots\}$, where $\mu_{1}+\mu_{2}+\mu_{3}=1$. Denote the EWSAoI for randomized policies (in the long run) as $\mathbb{E}[J^{R}]$.

\subsection{Age Analysis}
We will first find the exact EWSAoI of the coded randomized policy. We start with $\Delta_i(m)$ derived in \eqref{equ: delta_i[m]}. Consider $i=1$.
The expectation of $\Delta_1(m)$ is 
\begin{align}
&\mathbb{E}[\Delta_1(m)|\vec{s}(1)]\!\\
=&\!\mathbb{E}\left[\frac{I_{1}^{2}(m)}{2}+D_1(m-1)I_{1}(m)-\frac{I_1(m)}{2}\Big|\vec{s}(1)\right]\!\\
=&\frac{\mathbb{E}\left[{I_{1}^{2}}\right]}{2}+\mathbb{E}[D_1]\mathbb{E}[I_{1}]-\frac{\mathbb{E}[I_1]}{2}\label{equ: Randomized policy-(a)}.
\end{align}
Equality \eqref{equ: Randomized policy-(a)} holds because of the following observations:  (i) the processes $\{I_1(m)\}_m$ and $\{D_1(m)\}_m$ are each iid and not dependent on $\vec{s}(1)$ (so we use $I_1$ and $D_1$ to denote the underlying random variables, respectively), and (ii) $I_{1}(m)$ and $D_1(m-1)$ are independent (while $I_1(m)$ and $D_1(m)$ may be dependent). Thus, the EWSAoI of user $1$ is
\begin{align}
&\lim_{T\to\infty}\frac{1}{T}\sum_{k=1}^T \mathbb{E}[h_1(k)|\vec{s}(1)]\\
=&\lim_{T\to\infty}\frac{1}{T}\sum_{m=1}^{N_1(T)} \mathbb{E}[\Delta_1(m)|\vec{s}(1)]\\
=&\lim_{T\to\infty}\frac{{N_1(T)}}{T}\left(\frac{\mathbb{E}\left[{I_{1}^{2}}\right]}{2}+\mathbb{E}[D_1]\mathbb{E}[I_{1}]-\frac{\mathbb{E}[I_1]}{2}\right)\\
=&\frac{\mathbb{E}\left[I_{1}^{2}\right]}{2\mathbb{E}[I_1]}+\mathbb{E}[D_1]-\frac{1}{2}\label{eq:ageformularandom}
\end{align} 
where \eqref{eq:ageformularandom} holds because the arrival process is a renewal process \cite[Section 3.3, Theorem 3.3.1]{2008-SheldonM.Ross} and hence $\lim_{T\rightarrow\infty}\frac{T}{N_1(T)}=\mathbb{E}[I_1]$.

Then, we consider the statistics of $D_{1}$, and
 find the probability of $\Pr(D_1=d)$ for $d=1,2,\ldots$. For each slot~$k$, 
\begin{equation}\label{equ: Randomized p(h_1=d)}
\begin{aligned}
P(h_{1}(k)=1)&=\mu_{1}(1-\epsilon_{1})\\
P(h_{1}(k)=d)&=\mu_{1}\mu_{3}(1-\epsilon_{1})(\epsilon_{1}-\epsilon_{12})\\
&\times
(\mu_{1}\epsilon_{12}+\mu_{2}+\mu_{3}\epsilon_{1})^{d-2}\quad d\geq 2.
\end{aligned} 
\end{equation}
To find the probability distribution of $D_1$, condition the above probabilities on the event that a packet is delivered to user $1$ at time slot $k$. The probability of this event can be found by summing \eqref{equ: Randomized p(h_1=d)} over all $d\geq 1$:
$$P_{\text{delivery}}^{1}=\frac{\mu_1(1-\epsilon_1)(\mu_1+\mu_3)(1-\epsilon_{12})}{1-\mu_{1}\epsilon_{12}-\mu_{2}-\mu_{3}\epsilon_{1}}.$$
We thus find
\begin{equation*}
\footnotesize
\begin{aligned}
P(D_{1}=d)=\left\{\begin{array}{ll}
\mu_{1}(1-\epsilon_{i})/P_{\text{delivery}}^{1}&d=1\\
\frac{\mu_{1}\mu_{3}(\epsilon_{1}-\epsilon_{12})(1-\epsilon_{1})(\mu_{1}\epsilon_{12}+\mu_{2}+\mu_{3}\epsilon_{1})^{d-2}}{P_{\text{delivery}}^{1}}& d\geq2
\end{array}\right.
\end{aligned}
\end{equation*}
and the expectation of $D_1$ is equal to
\begin{equation}\label{equ: R-ED}
\footnotesize
\begin{aligned}
E(D_1)=1+\frac{\mu_3(\epsilon_{1}-\epsilon_{12})}{(\mu_1+\mu_3)(1-\epsilon_{12})(1-\mu_{1}\epsilon_{12}-\mu_{2}-\mu_{3}\epsilon_{1})}.
\end{aligned}
\end{equation}

The distribution of $I_1$ can be found by treating $I_1$ and $D_1$ jointly. First of all, we have $$P(I_{1}=1)=P(I_{1}=1,D_1=1)=\mu_{1}(1-\epsilon_{1}).$$ 
Then we look at the event of $I_1=\ell$ and $D_1=d$ for $\ell\geq 2, d\leq\ell$ (otherwise, if $d>\ell$, then $\Pr(I_1=\ell, D_1=d)=0$ because  the packet in the queue $Q_2^{(1)}$ becomes obsolete once a new packet is delivered to user $1$).  
So we assume $d\leq\ell$ and consider the following cases: (\romannumeral1) If $d=1$, then  a packet was delivered by action $i$ at slot $\ell$. 
(\romannumeral2) For $d\geq2$, the delivered packet was moved to $Q_2^{(1)}$ at slot $\ell-d+1$, stayed there, and got received at user $1$ at slot $\ell$. Now  consider slots $1$ to  $\ell-d+1$. Denote by $t$ the first slot in which a packet is received in $Q_{2}^{(1)}$, $1\leq t\leq \ell-d+1$. Then we have the two sub-cases: 
(\romannumeral2-1) If $t$ exists,
then $Q_2^{(1)}$ is empty before $t$, and the delivered packet (another packet different from the delivered packet) moves to $Q_2^{(1)}$ at $t$ when $t=\ell-d+1$ (when $t<\ell-d+1$) and from that point $Q_2^{(1)}$ is non-empty. (\romannumeral2-2) If there is no such slot $t$, which may happen for $d=1$, then $Q_2^{(1)}$ is empty for the entire duration of $\ell$ slots. Considering the above cases, for $\ell=1,2,\ldots$,  we can show the following lemma. 
\begin{lemma}\label{lem: distribution of I_i}
The probability distribution of the inter delivery random variable $I_{1}$ is given by
\begin{equation}\label{equ: dirstibution of I_i}
P(I_{1}=\ell)=\delta_{1} x_{1}^{\ell-1}+\beta_{1} y_{1}^{\ell-1}
\end{equation} 
where $$x_{1}=\mu_1\epsilon_{12}+\mu_{2}+\mu_3$$ $$y_{1}=\mu_1\epsilon_{1}+\mu_{2}+\mu_3\epsilon_1$$ $$\delta_{1}=\mu_1(1-\epsilon_{12})+\frac{\mu_{1}^{2}(\epsilon_{1}-\epsilon_{12})(1-\epsilon_{12})}{-\mu_{1}(\epsilon_{1}-\epsilon_{12})+\mu_{3}(1-\epsilon_{1})}$$ 
$$\beta_{1}=-\frac{\mu_{1}(\epsilon_{1}-\epsilon_{12})(1-\epsilon_1)}{-\mu_{1}(\epsilon_{1}-\epsilon_{12})+\mu_{3}(1-\epsilon_{1})}\Big({\mu_{1}}+{\mu_3}\Big).$$
\end{lemma}
The proof of Lemma~\ref{lem: distribution of I_i} is in Appendix~\ref{App: Lemma 1}.
Using Lemma~\ref{lem: distribution of I_i}, we find $\mathbb{E}[I_{1}]$ and $\mathbb{E}[I_{1}^2]$:
\begin{align}
\mathbb{E}[I_{1}]=&\frac{\delta_{1}}{(1-x_{1})^2}+\frac{\beta_{1}}{(1-y_{1})^2}\label{eq:exp}\\
\mathbb{E}[I_{1}^2]=&\frac{\delta_{1}(1+x_{1})}{(1-x_{1})^3}+\frac{\beta_{1}(1+y_{1})}{(1-y_{1})^3}.\label{eq:second}
\end{align}
Finally, substituting \eqref{equ: R-ED}, \eqref{eq:exp},  \eqref{eq:second} into \eqref{eq:ageformularandom}, and replacing for the values of $x_i,y_i,z_i$ and $\delta_i, \beta_i$, we find the EWSAoI as given by the following theorem.

\begin{theorem}\label{thm: R-age} 
The $EWSAoI$ of Randomized policy is characterized by 
\begin{equation}\label{equ: R-randomizedpolicyage}
\footnotesize
\begin{aligned}
&\mathbb{E}[J^{R}]\!=\!\frac{1}{2}\!\sum_{i=1}^{2}\!\alpha_{i}\!\!\left(\!\!\frac{\frac{\mu_3(1-\epsilon_i)}{\left(\mu_i(1-\epsilon_{12})\right)^2}+\frac{-\mu_i(\epsilon_{i}-\epsilon_{12})}{\left((\mu_i+\mu_3)(1-\epsilon_i)\right)^2}}{\frac{\mu_3(1-\epsilon_i)}{\mu_i(1-\epsilon_{12})}+\frac{-\mu_i(\epsilon_{i}-\epsilon_{12})}{(\mu_i+\mu_3)(1-\epsilon_i)}}\right.\\
&\left.\qquad\qquad\qquad\  \ \ +\frac{\mu_3(\epsilon_{i}-\epsilon_{12})}{(\mu_i\!+\!\mu_3)\!(1\!-\!\epsilon_{12})\!(\mu_i(1\!-\!\epsilon_{12})\!\hspace{-.05cm}+\hspace{-.05cm}\!\mu_3(1\!-\!\epsilon_i)\!)}\!\!\right)\!\hspace{-.05cm}.
\end{aligned}
\end{equation}
\end{theorem}

\begin{remark}
To find an \emph{optimal coded randomized policy} with respect to age, we have to choose the probability vector $(\mu_{1}^{*}, \mu_{2}^{*}, \mu_{3}^{*})$ such that
$\mathbb{E}[J^{R}]$ is minimized. 
\end{remark} 
\begin{remark}
Setting $\mu_3=0$, we recover the EWSAoI of \cite{2018-IgorKadota} which corresponds to uncoded randomized policies. 
\end{remark}

\subsection{Symmetric BPECs}

Consider the class of symmetric BPECs. Let the erasure probabilities of channels to users $1$ and $2$ be equal to $\epsilon$ and  $\epsilon_{12}$ be the probability of simultaneous erasure at both users. So $\epsilon>\epsilon_{12}$. Note that $\epsilon_{12}$ is either  a function of $\epsilon$ or a constant, thus we rewrite $\epsilon_{12}$ as $\epsilon_{12}(\epsilon)$. Denote the probabilities of choosing action $1$ and $2$ by $\mu$, hence the probability of choosing action $3$ is $1-2\mu$ with $0\leq\mu\leq1/2$. For simplicity, let $\alpha_1=\alpha_2=\alpha$. We find regimes of operation where optimal coded randomized policies strictly improve the EWSAoI over  uncoded randomized policies such as \cite{2018-IgorKadota}. 

From Theorem \ref{thm: R-age}, we find the EWSAoI as follows
\begin{equation}\label{equ: R-symmetric AoI}
\footnotesize
\begin{aligned}
\mathbb{E}[J^{R}]=&\frac{\alpha}{(1-\epsilon_{12}(\epsilon))(1-\mu)}\\
\times&\Bigg(\frac{(1-\epsilon)^{3}(1-2\mu)(1-\mu)^{2}-(1-\epsilon_{12}(\epsilon))^{2}(\epsilon-\epsilon_{12}(\epsilon))\mu^{3}}{\Big((1-2\mu)(1-\epsilon)^{2}(1-\mu)-\mu^{2}(1-\epsilon_{12}(\epsilon))(\epsilon-\epsilon_{12}(\epsilon))\Big)\mu}\\
+&\frac{(1-2\mu)(\epsilon-\epsilon_{12}(\epsilon))}{(2\epsilon-1-\epsilon_{12}(\epsilon))\mu+1-\epsilon}\Bigg).
\end{aligned}
\end{equation}
First, we get optimal randomized policies (i.e., optimal probabilities) of the symmetric case by the following lemma.
\begin{lemma}\label{lem: symmetric case}
For any given $\epsilon$, 
\begin{equation}\label{equ: symmetric case}
\footnotesize
\mu^{*}=\left\{\begin{aligned}
&\frac{\sqrt{1-\epsilon}}{\sqrt{1-\epsilon}+\sqrt{\epsilon-\epsilon_{12}(\epsilon)}}&\quad& \epsilon_{12}(\epsilon)-2\epsilon+1<0\\
&1/2&\quad&\epsilon_{12}(\epsilon)-2\epsilon+1\geq0.
\end{aligned}\right.
\end{equation}
 \end{lemma}
The proof of Lemma~\ref{lem: symmetric case} is in Appendix~\ref{App: Lemma 3}. 
From Lemma~\ref{lem: symmetric case}, if $\epsilon_{12}(\epsilon)-2\epsilon+1\geq0$, then the optimal probability of action $1$ (action $2$) is $\mu^{*}=1/2$. In this case, no packets reach the XOR, i.e., action $3$ is never chosen. So coded randomized policies degenerate to uncoded randomized algorithms. Thus coded and uncoded randomized policies reach the same age.
If $\epsilon_{12}(\epsilon)-2\epsilon+1<0$, since $\mu^{*}$ is the {\it unique} minimum point of \eqref{equ: R-symmetric AoI} as shown in Appendix~\ref{App: Lemma 3}, then 
\begin{align*}
 \mathbb{E}[J^{R}]\big|_{\{\mu^{*}\}}< \mathbb{E}[J^{R}]\big|_{\{\mu=1/2\}}
\end{align*}
which implies that the age of optimal coded randomized policies is strictly lower than that of optimal uncoded randomized policies. So we arrive at the following result.
\begin{theorem}\label{thm: symmetric case}
Optimal coded and uncoded randomized policies achieve the same EWSAoI if and only if $\epsilon_{12}(\epsilon)-2\epsilon+1\geq 0$. Otherwise,  optimal coded randomized policies strictly outperform uncoded randomized policies. 
\end{theorem}
\begin{remark}
When the channels are independent, i.e., $\epsilon_{12}(\epsilon)=\epsilon^{2}$, coded and uncoded randomized policies have the same performance with respect to age.
\end{remark}

\section{Max-Weight Policies}\label{sec: Standard Max-Weight Policies}
\subsection{The Algorithm}
In this section, we devise deterministic policies using techniques from Lyapunov Optimization. Denote the EWSAoI for Max-Weight policies (in the long run) as $\mathbb{E}[J^{MW}]$. Denote $\vec{h}(k)=(h_{1}(k),h_{2}(k))$ and $\vec{s}(k)=(h_{1}(k),h_{2}(k),w_{1}(k),w_{2}(k))$.
Define the Lyapunov function
\begin{align}
L(\vec{h}(k))=\frac{1}{2}\sum_{i=1}^{2}\alpha_{i}h_{i}^{2}(k),
\end{align}
and the one-slot Lyapunov Drift 
\begin{align}
\Theta(\vec{h}(k))=\mathbb{E}[L(\vec{h}(k+1))-L(\vec{h}(k))|\vec{s}(k)].
\end{align}

We  devise the Max-Weight (MW) policy such that it minimizes the one-slot Lyapunov drift:
\begin{definition}\label{def: MW-policy}
In each slot $k$, the MW policy chooses the action that has the maximum weight as shown in following Table:
\begin{figure}[h!]
\begin{tabular}{|c|c|}
\hline 
 $\!\!\!A(k)\!\!\!$& Weights\\
\hline 
$1$&$\frac{\alpha_{1}(1-\epsilon_{1})}{2}h_{1}(k)(h_{1}(k)+2)$ \\
\hline 
$2$&$\frac{\alpha_{2}(1-\epsilon_{2})}{2}h_{2}(k)(h_{2}(k)+2)$ \\
\hline
$3$&$\!\!\frac{1}{2}\sum_i\!\alpha_{i}(\hspace{-.05cm}1\!\hspace{-.05cm}-\hspace{-.05cm}\!\epsilon_{i}\hspace{-.05cm})1_{\!\{w_{i}(k)\!>\!0\}}\hspace{-.05cm}(h_{i}^{2}(k)\!+\!2h_{i}(k)\!\hspace{-.05cm}-\hspace{-.05cm}\!w_{i}^{2}(k)\!\hspace{-.05cm}-\hspace{-.05cm}\!2w_{i}(k)\hspace{-.05cm})\!\!$\\
\hline
\end{tabular}
\end{figure}	
\end{definition}
\begin{theorem}\label{thm: MW-policy}
The MW policy defined in Definition \ref{def: MW-policy} minimizes the one-slot Lyapunov Drift in each slot.
\end{theorem}
\begin{proof}
The proof of Theorem \ref{thm: MW-policy} can be found in  Appendix~\ref{App: Theorem 4}.
\end{proof}

\subsection{The Upper Bound of EWSAoI}
While the exact analysis of the resulting EWSAoI is difficult for the above MW policy, we derive an upper bound on it in Theorem \ref{thm: MW-the upper bound}. The proof is deferred to Appendix~\ref{App: Theorem 5}.

\begin{theorem}\label{thm: MW-the upper bound}
The EWSAoI achieved by the proposed Max-Weight policy is upper bounded by
\begin{align*}
\mathbb{E}[J^{MW}]\leq\sqrt{\frac{1}{2}\sum_{i=1}^{2}\frac{\alpha_{i}}{\mu_{i}(1-\epsilon_{i})}\sum_{i=1}^{2}\alpha_{i}\Psi_{i}}+\frac{1}{2}\sum_{i=1}^{2}\alpha_{i}\Phi_{i}	
\end{align*}
where 
$$\Phi_{i}=\frac{1-\mu_{3}P_{ne}^{i}(1-\epsilon_{i})-\mu_{i}(1-\epsilon_{i})}{\mu_{i}(1-\epsilon_{i})}$$
$$\Psi_{i}=1-\mu_{3}P_{ne}^{i}(1-\epsilon_{i})+\frac{\big(1-\mu_{i}(1-\epsilon_{i})-\mu_{3}P_{ne}^{i}(1-\epsilon_{i})\big)^{2}}{\mu_{i}(1-\epsilon_{i})}$$
$$P_{ne}^{i}=\frac{\mu_{i}(\epsilon_{i}-\epsilon_{12})}{\mu_{3}(1-\epsilon_{i})+\mu_{i}(\epsilon_{i}-\epsilon_{12})-\mu_{3}\mu_{i}(1-\epsilon_{i})(\epsilon_{i}-\epsilon_{12})}.$$
\end{theorem}
\begin{remark}
Setting $\mu_3=0$, we recover the upper bound of \cite{2018-IgorKadota} which corresponds to uncoded randomized policies.
\end{remark}

\section{Numerical Results and Discussion}\label{sec: Numerical results and discussions}

\begin{figure}[ht!]
\begin{tikzpicture}
\begin{axis}
[axis lines=left,
width=2.8in,
height=2.5in,
scale only axis,
xlabel=$\epsilon$,
ylabel=EWSAoI,
xmin=0.5,xmax=.9,
ymin=0,ymax=8,
xtick={0.1,0.3,0.5,0.7,0.9},
ytick={0,4,8,12,16,20},
legend pos=north west,
ymajorgrids=true,
grid style=dashed,
scatter/classes={
a={mark=+, draw=black},
b={mark=star, draw=black}
}
]

\addplot[smooth,color=black]
coordinates{(0.1,0.7563)(0.15,0.7727)(0.2,0.7914)(0.25,0.8130)(0.3,0.8380)(0.35,0.8672)(0.4,0.9016)(0.45,0.9428)(0.5,0.9926)(0.55,1.0541)(0.6, 1.1316)(0.65,1.2319)(0.7, 1.3664)(0.75,1.5556)(0.8,1.8407)(0.85,2.3174)(0.9,3.2733)(0.95,5.2809)
};

\addplot[
  color=red,
  mark=star,]
coordinates{(0.1,1.0647)(0.15,1.1274)(0.2,1.1978)(0.25,1.2777)(0.3,1.3689)(0.35,1.4742)(0.4,1.5971)(0.45,1.7423)(0.5,1.9165)(0.55,2.1286)(0.6,2.3847)(0.65,2.7011)(0.7,3.1060)(0.75,3.6491)(0.8,4.4266)(0.85,5.6567)(0.9,7.9711)(0.95,13.9950)
};

\addplot[
  color=green,
  mark=square,]
coordinates{(0.1,1.0647)(0.15,1.1274)(0.2,1.1978)(0.25,1.2777)(0.3,1.3689)(0.35,1.4742)(0.4,1.5971)(0.45,1.7423)(0.5,1.9165)(0.55,2.1295)(0.6,2.3956)(0.65,2.7379)(0.7,3.1942)(0.75,3.8330)(0.8,4.7913)(0.85,6.3884)(0.9,9.5826)(0.95,19.1652)
};

\addplot[
color=blue,dashed,
mark=triangle]
coordinates{(0.1,0.8275)(0.15,0.8749)(0.2,0.9285)(0.25,0.9892)(0.3,1.0590)(0.35,1.1390)(0.4,1.2306)(0.45,1.3398)(0.5,1.4692)(0.55,1.6252)(0.6,1.8204)(0.65,2.0697)(0.7,2.3985)(0.75,2.8586)(0.8,3.5473)(0.85,4.6932)(0.9,6.9798)(0.95,13.9734)
};

\addplot[color=purple, mark=o]
coordinates{(0.1,0.8265)(0.15,0.8722)(0.2,0.9241)(0.25,0.9834)(0.3,1.0518)(0.35,1.1304)(0.4,1.2213)(0.45,1.3312)(0.5,1.4619)(0.55,1.6210)(0.6,1.8203)(0.65,2.0778)(0.7,2.4195)(0.75,2.8972)(0.8,3.6165)(0.85,4.8137)(0.9,7.2164)(0.95,14.5141)
};

%
%
%
%

\legend{Lower Bound, Coded Randomized, Uncoded Randomized, Coded Max-Weight, Uncoded Max-Weight}
\end{axis}
\end{tikzpicture}
\caption{EWSAoI as a function of erasure probability for a class of dependent channels with $\epsilon_{12}=\epsilon^{2}/5$.}
\label{fig-age1}
\end{figure}
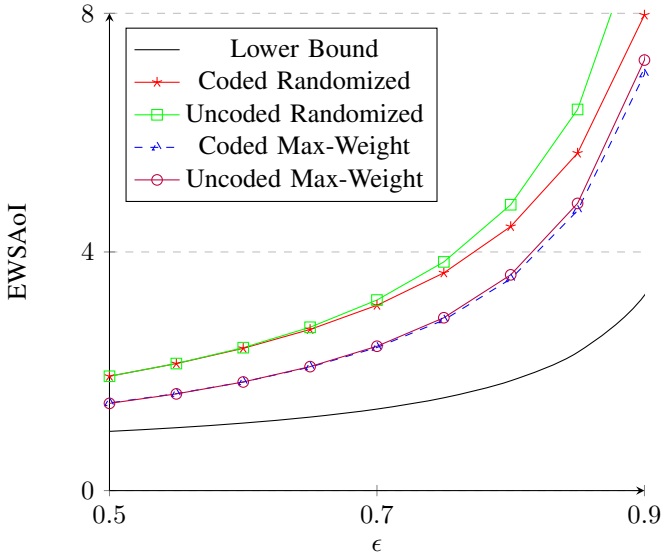

In this section, we compare the performance of the proposed coded algorithms with the uncoded algorithms in \cite{2018-IgorKadota}. In Figure \ref{fig-age1}, we plot the EWSAoI of optimal coded and uncoded randomized policies, coded and uncoded Max-Weight policies, as well as the lower bound of Theorem \ref{thm: lowerbound}. 
We have chosen $\alpha_{1}=0.3$, $\alpha_{2}=0.7$, and $\epsilon_{1}=\epsilon_{2}=\epsilon$, where $\epsilon$ varies from $0.5$ to $0.9$. We consider a dependent channel of $\epsilon_{12}=\epsilon^{2}/5$. We see that for the proposed policies, coding is beneficial when the channel erasure is larger than a threshold. Although the gain is  small for MW policies, we believe the gain will be  more significant over networks with many users.
In addition, we also consider Dynamic Programming (DP) policies in the class of 3-action coding schemes, the
EWSAoI performance of DP policies is very close to that of MW policies for all $\epsilon$.

\appendices

\section{Proof of Lemma~\ref{lem: buffer size 1}}\label{App: Lemma buffer size 1}
\begin{proof}
We consider two types of policies: policies with buffer size $1$, denote by $\pi_{1}$, and policies with buffer size larger than $1$, denote by $\pi_{2}$.
$\pi_1$ works on buffers of size $1$ while $\pi_{2}$ works on buffers of infinite size. To differentiate the two policies and their corresponding queues, we label the packets inside the queues by \emph{new} and \emph{old}. A new packet in a queue means that this packet has not been previously received at its intended user. A packet in a queue is labeled old if there is a newer packet in the same queue or if the packet (or a fresher packet) is already received at its intended user. We refer to the freshest old packet as the old packet. At a given time slot, denote the new packet and the  old packet for user $1$ (resp. user $2$) as $p_{new}$ and $p_{old}$ (resp. $q_{new}$ and $q_{old}$) in $Q_{1}^{(i)}$ (resp. $Q_{2}^{(i)}$), $i=1,2$, respectively.

Consider a fixed erasure pattern for the channels to the users and an arbitrary time slot $k$. Denote the generation time of the new packet and the  old packet for user $1$ (resp. user $2$) in $Q_{1}^{(i)}$ (resp. $Q_{2}^{(i)}$) as $k_{pn}$ and $k_{po}$ (resp. $k_{qn}$ and $k_{qo}$), $i=1,2$.
We measure the efficiency of the policies at each slot $k$ by their resulting age at both users. To this end, we show that no matter what policy $\pi_2$ chooses, there is a policy of type $\pi_1$ that is at least as efficient as $\pi_2$ in terms of age at both users.

\begin{enumerate}
\item $\pi_{2}$ chooses action $i$, $i=1,2$, at slot $k$.
In this case, $\pi_{2}$ transmits a packet from $Q_{1}^{(i)}$. By sending the new packet from $Q_{1}^{(i)}$, policy $\pi_1$ attains a smaller age for user $i$ at slot $k+1$ and the age remains the same for the other user. See also \cite[Theorem 2]{2019-E.Najm} for a similar argument.
\item $\pi_{2}$ chooses action $3$ at slot $k$. First of all, note that it is always strictly better to choose action $i$ over action $3$ if queue $Q_2^{(\backslash i)}$ is empty, and then a similar analysis with the previous case can be done. Therefore, we consider the case that both queues $Q_2^{(1)}$ and $Q_2^{(2)}$ are non-empty and a coded packet of the form $p_{old}\oplus q_{old}$ is sent by $\pi_2$. To compare with policies of type $\pi_1$, we need to consider four possible cases depending on if there are any new packets in $Q_2^{(1)}$ and $Q_2^{(2)}$. Note that type $\pi_1$ policies are effectively oblivious to old packets.
\begin{itemize}
\item There are packets of type $p_{old}$ as well as $p_{new}$ in  $Q_{1}^{(2)}$ and similarly $q_{old},q_{new}$ in $Q_{2}^{(2)}$. In this case, $\pi_1$ can choose action $3$ and send the coded packet $p_{new}\oplus q_{new}$.  This way, if there is a delivery at user~$1$ (resp. user~$2$), the age of user~$1$ (resp. user~$2$) will drop to $k+1-k_{po}$ (resp. $k+1-k_{qo}$) under $\pi_2$ and to $k+1-k_{pn}$ (resp. $k+1-k_{qn}$) under $\pi_1$ at time $k+1$. Since $k_{pn}>k_{po}$ and $k_{qn}>k_{qo}$, $\pi_1$ is strictly more efficient than $\pi_2$ at slot $k+1$.
\item There are no packets of type $p_{new}$ in $Q_{1}^{(2)}$ but  there are packets of type $q_{old}$ as well as $q_{new}$ in $Q_{2}^{(2)}$. Since there is no packets of type new in $Q_1^{(2)}$,  the existing old packet (or a fresher packet) must have been previously received at user $1$. So the current age at user~$1$ is at most $k-k_{po}$ and will increase at most up to $k+1-k_{po}$ at slot $k+1$. The best that a policy of type $\pi_2$ can do by action $3$ is to drop the age at user~$1$ to $k+1-k_{po}$ and drop the age at user~$2$ to $k+1-k_{qo}$ at slot $k+1$. $\pi_1$ can choose action $2$ and perform strictly better than $\pi_2$ at slot $k+1$ because the newest packet in $Q_1^{(2)}$ is fresher than any new packet in $Q_2^{(2)}$.
\item There are no packets of type $q_{new}$ in $Q_{2}^{(2)}$ but  there are packets of type $p_{old}$ as well as $p_{new}$ in $Q_{2}^{(1)}$. This case is similar to the previous case.
\item There are only packets of type $p_{old}$ in $Q_{1}^{(2)}$ and  $q_{old}$ in $Q_{2}^{(2)}$. Since there is no packet of type new in the queues, it means that the old packets (or fresher ones) are previously received at the users under $\pi_{1}$. So under $\pi_1$, the current age at user $1$ (resp. user $2$) is at most $k-k_{po}$ (resp. $k-k_{qo}$) at this time and will increase at most by $1$ at time $k+1$. The age under $\pi_2$, however, can only become $k+1-k_{po}$ (resp. $k+1-k_{qo}$) if there is a delivery. Hence $\pi_1$ is at least as efficient as $\pi_2$.
\end{itemize}
\end{enumerate}
So at each time, the  average age of policy $\pi_{1}$ is smaller than or equal to that of policy $\pi_{2}$. 

\end{proof}

\section{Proof of Lemma~\ref{lem: distribution of I_i}}\label{App: Lemma 1}
\begin{proof}
From the analysis above,
 the probability of Case (\romannumeral1) is $$P_{1}=(\mu_{1}\epsilon_{12}+\mu_{2}+\mu_{3})^{\ell-1}\mu_1(1-\epsilon_1).$$
 The probability of Case (\romannumeral2-1) is
\begin{equation*}
\footnotesize
\begin{aligned}
 P_{2}=&\sum_{t=1}^{\ell-1}(\mu_{1}\epsilon_{12}+\mu_{2}+\mu_{3})^{t-1}(\mu_{1}\epsilon_{1}+\mu_{2}+\mu_{3}\epsilon_{1})^{\ell-1-t}\\
 \times&\mu_{1}(\epsilon_{1}-\epsilon_{12})\mu_1(1-\epsilon_1)\\
+&\sum_{2\leq d\leq\ell}\sum_{t=1}^{\ell-d}\Big((\mu_{1}\epsilon_{12}+\mu_{2}+\mu_{3})^{t-1}(\mu_{1}\epsilon_{1}+\mu_{2}+\mu_{3}\epsilon_{1})^{\ell-d-t}\\
\times&(\mu_{1}(\epsilon_{1}-\epsilon_{12}))^2(\mu_{2}+\mu_1\epsilon_{12}+\mu_3\epsilon_1)^{d-2}\mu_3(1-\epsilon_1)\Big).
\end{aligned}
\end{equation*}
The probability of Case (\romannumeral2-2) is
\begin{equation*}
\begin{aligned}
P_{3}=&\sum_{2\leq d\leq\ell}(\mu_{1}\epsilon_{12}+\mu_{2}+\mu_{3})^{\ell-d}\\
\times&\mu_{1}(\epsilon_{1}-\epsilon_{12})(\mu_{2}+\mu_1\epsilon_{12}+\mu_3\epsilon_1)^{d-2}\mu_3(1-\epsilon_1).
\end{aligned}
\end{equation*}
Denote \begin{align*}
x_{1}=&\mu_1\epsilon_{12}+\mu_{2}+\mu_3\\
y_{1}=&\mu_1\epsilon_{1}+\mu_{2}+\mu_3\epsilon_1\\
z_{1}=&\mu_1\epsilon_{12}+\mu_{2}+\mu_3\epsilon_1,
\end{align*}
thus
\begin{align*}
&\Pr(I_1=\ell)\nonumber=\sum_{d\leq\ell}\Pr(I_{1}=\ell, D_1=d)\nonumber
\\=&P_{1}+P_{2}+P_{3}=eq_{1}+eq_{2}+eq_{3}+eq_{4}.
\end{align*}
where
\begin{equation*}
\begin{aligned}
eq_{1}=&x_{1}^{\ell-1}\mu_1(1-\epsilon_1)\\
eq_{2}=&\mu_{1}^{2}(\epsilon_{1}-\epsilon_{12})(1-\epsilon_1)\frac{x_{1}^{l-1}-y_{1}^{l-1}}{x_{1}-y_{1}}\\
eq_{3}=&\frac{(\mu_{1}(\epsilon_{1}-\epsilon_{12}))^2\mu_3(1-\epsilon_1)}{x_{1}-y_{1}}\\
\times&\Big(\frac{x_{1}^{l-1}-z_{1}^{l-1}}{x_{1}-z_{1}}-\frac{y_{1}^{l-1}-z_{1}^{l-1}}{y_{1}-z_{1}}\Big)\\
eq_{4}=&\mu_{1}(\epsilon_{1}-\epsilon_{12})\mu_3(1-\epsilon_1)\frac{x_{1}^{l-1}-z_{1}^{l-1}}{x_{1}-z_{1}}.
\end{aligned}
\end{equation*}
By calculation, for $\ell\geq2$,
\begin{align}\label{equ: probability of I}
\Pr(I_1=\ell)=\delta_{1} x_{1}^{\ell-1}+\beta_{1} y_{1}^{\ell-1}
\end{align}
where 
$$\delta_{1}=\mu_1(1-\epsilon_{12})+\frac{\mu_{1}^{2}(\epsilon_{1}-\epsilon_{12})(1-\epsilon_{12})}{-\mu_{1}(\epsilon_{1}-\epsilon_{12})+\mu_{3}(1-\epsilon_{1})}$$
$$\beta_{1}=-\frac{\mu_{1}(\epsilon_{1}-\epsilon_{12})(1-\epsilon_1)}{-\mu_{1}(\epsilon_{1}-\epsilon_{12})+\mu_{3}(1-\epsilon_{1})}\Big({\mu_{1}}+{\mu_3}\Big).$$
Now we check $\Pr(\ell=1)$. Substituting $\ell=1$ into \eqref{equ: probability of I}, we have 
$P(I_{1}=\ell)=\delta_{1}+\beta_{1}=\mu_{1}(1-\epsilon_{1})$, therefore 
\begin{equation*}
P(I_{1}=\ell)=\delta_{1} x_{1}^{\ell-1}+\beta_{1} y_{1}^{\ell-1}
\end{equation*}
holds for all $\ell\geq1$.
\end{proof}

\section{Proof of Lemma \ref{lem: symmetric case}}\label{App: Lemma 3}
\begin{proof}
Taking the first derivative of \eqref{equ: R-symmetric AoI} with respect to $\mu$ and letting it equal to zero, we have
\begin{align*}
&\mu^{(1)}=\frac{1-\epsilon+\sqrt{(\epsilon-\epsilon_{12}(\epsilon))(1-\epsilon)}}{\epsilon_{12}(\epsilon)-2\epsilon+1}\\
&\mu^{(2)}=\frac{1-\epsilon-\sqrt{(\epsilon-\epsilon_{12}(\epsilon))(1-\epsilon)}}{\epsilon_{12}(\epsilon)-2\epsilon+1}.
\end{align*}
If $\epsilon_{12}(\epsilon)-2\epsilon+1=0$, then 
substituting $\epsilon_{12}(\epsilon)=2\epsilon-1$ into \eqref{equ: R-symmetric AoI}, taking the derivative and equating to zero, we obtain $\mu^{(1)}=\mu^{(2)}=1/2$.

First, we simplify $\mu^{(2)}$. If $\epsilon_{12}(\epsilon)-2\epsilon+1\neq0$, then
\begin{align}
\mu^{(2)}=&\frac{1-\epsilon-\sqrt{(\epsilon-\epsilon_{12}(\epsilon))(1-\epsilon)}}{\epsilon_{12}(\epsilon)-2\epsilon+1}\\=&\frac{\sqrt{1-\epsilon}}{\epsilon_{12}(\epsilon)-2\epsilon+1}(\sqrt{1-\epsilon}-\sqrt{\epsilon-\epsilon_{12}(\epsilon)})\\
=&\frac{\sqrt{1-\epsilon}}{\epsilon_{12}(\epsilon)-2\epsilon+1}\cdot\frac{\epsilon_{12}(\epsilon)-2\epsilon+1}{\sqrt{1-\epsilon}+\sqrt{\epsilon-\epsilon_{12}(\epsilon)}}\\=&\frac{\sqrt{1-\epsilon}}{\sqrt{1-\epsilon}+\sqrt{\epsilon-\epsilon_{12}(\epsilon)}}\label{equ: simplify x_2}.
\end{align}
Note that \eqref{equ: simplify x_2} still holds when $\epsilon_{12}(\epsilon)-2\epsilon+1=0$. From \eqref{equ: simplify x_2}, we have $\mu^{(2)}>0$.

Consider the following three cases: 

(1) If $\epsilon_{12}(\epsilon)-2\epsilon+1<0$, then $\mu^{(1)}<0$. Also, 
\begin{align*}
\mu^{(2)}=&\frac{\sqrt{1-\epsilon}}{\sqrt{1-\epsilon}+\sqrt{\epsilon-\epsilon_{12}(\epsilon)}}\\<&\frac{\sqrt{1-\epsilon}}{\sqrt{1-\epsilon}+\sqrt{1-\epsilon}}<1/2.
	\end{align*}
$\mathbb{E}[J^{R}]$ in \eqref{equ: R-symmetric AoI} decreases when $\mu\in[0,\mu^{(2)}]$ and increases when $\mu\in(\mu^{(2)},1/2]$, so the unique minimal point is obtained at $\mu^{*}=\mu^{(2)}$.  

(2) If $\epsilon_{12}(\epsilon)-2\epsilon+1=0$, then  $\mu^{(1)}=\mu^{(2)}=1/2$, and $\mathbb{E}[J^{R}]$ in \eqref{equ: R-symmetric AoI} decreases when $\mu\in[0, 1/2]$. So the unique minimum point is obtained at $\mu^{*}=1/2$.

(3) If $\epsilon_{12}(\epsilon)-2\epsilon+1>0$, then $$\mu^{(1)}>\mu^{(2)}>0.$$ Also, 
\begin{align*}
\mu^{(2)}=&\frac{\sqrt{1-\epsilon}}{\sqrt{1-\epsilon}+\sqrt{\epsilon-\epsilon_{12}(\epsilon)}}\\>&\frac{\sqrt{1-\epsilon}}{\sqrt{1-\epsilon}+\sqrt{1-\epsilon}}>1/2.
\end{align*}
$\mathbb{E}[J^{R}]$ decreases when $\mu\in[0,1/2]$, so the unique minimum point is obtained at $\mu^{*}=1/2$. 
\end{proof}

\section{Proof of Theorem \ref{thm: MW-policy}}\label{App: Theorem 4}
\begin{proof}
We first denote by $d_{i}(k)\in\left\{0,1\right\}$ and $t_{i}(k)\in\left\{0,1\right\}$, the number of packets delivered to user $i$ during slot $k$ from $Q_{1}^{(i)}$ and $Q_{2}^{(i)}$, respectively. We have
\begin{align*}
&d_{i}(k)\leq1\\
&t_{i}(k)\leq1\\
&d_{i}(k)+t_{i}(k)\leq1.
\end{align*} 
Thus from \eqref{equ: BM-w} and \eqref{equ: BM-age},
\begin{equation}\label{equ: MW-recursion of h}
\begin{aligned}
h_{i}(k+1)=&d_{i}(k)+t_{i}(k)(w_{i}(k)+1)\\
+&(1-d_{i}(k)-t_{i}(k))(h_{i}(k)+1).
\end{aligned}
\end{equation}
Using \eqref{equ: MW-recursion of h}, we can re-write the Lyapunov Drift as follows:
\begin{equation*}
\footnotesize
\begin{aligned}
\Theta(\vec{h}(k))=&\mathbb{E}\big[L(\vec{h}(k+1))-L(\vec{h}(k))|\vec{s}(k)\big]
\\=&\frac{1}{2}\sum_{i=1}^{2}\alpha_{i}\mathbb{E}\big[\big((1-d_{i}(k)-t_{i}(k))(h_{i}(k)+1)\\
+&d_{i}(k)+t_{i}(k)(w_{i}(k)+1)\big)^{2}-h_{i}^{2}(k)|\vec{s}(k)\big]\\
\stackrel{(a)}{=}&\frac{1}{2}\sum_{i=1}^{2}\alpha_{i}\mathbb{E}\big[\big((1-d_{i}(k)-t_{i}(k))^{2}(h_{i}(k)+1)^{2}\\
+&d_{i}^{2}(k)+t_{i}^{2}(k)(w_{i}(k)+1)^{2}-h_{i}^{2}(k)|\vec{s}(k)\big)\big]\\
\stackrel{(b)}{=}&\frac{1}{2}\sum_{i=1}^{2}\alpha_{i}\mathbb{E}\big[\big((1-d_{i}(k)-t_{i}(k))\cdot(h_{i}(k)+1)^{2}\\
+&d_{i}(k)+t_{i}(k)(w_{i}(k)+1)^2-h_{i}^{2}(k)|\vec{s}(k)\big)\big]\\
=&\frac{1}{2}\sum_{i=1}^{2}\alpha_{i}\Big(\mathbb{E}[t_{i}(k)|\vec{h}(k)]\big(w_{i}^{2}(k)+2w_{i}(k)-h_{i}^{2}(k)-2h_{i}(k)\big)\\
-&\mathbb{E}[d_{i}(k)|\vec{h}(k)](h_{i}^{2}(k)+2h_{i}(k))+2h_{i}(k)+1\Big).
\end{aligned}
\end{equation*}
(a) holds because from the definition of $d_{i}(k)$ and $t_{i}(k)$, at each slot only one of $d_{i}(k), t_{i}(k), 1-d_{i}(k)-t_{i}(k)$ equals to $1$, and the rest are zero.
(b) holds because $d_{i}^{2}(k)=d_{i}(k)$, $t_{i}^{2}(k)=t_{i}(k)$, and $(1-d_{i}(k)-t_{i}(k))^{2}=1-t_{i}(k)-d_{i}(k)$.

Minimizing $\Theta(\vec{h}(k))$ is equivalent to minimize 
\begin{equation}\label{equ: MW-upper bound-tilde delta}
\footnotesize
\begin{aligned}
\tilde{\Theta}(\vec{h}(k))=&\frac{1}{2}\sum_{i=1}^{2}\alpha_{i}\Big(-\mathbb{E}[d_{i}(k)|\vec{s}(k)]\big(h_{i}^{2}(k)+2h_{i}(k)\big)\\
+&\mathbb{E}[t_{i}(k)|\vec{s}(k)]\big(w_{i}^{2}(k)+2w_{i}(k)-h_{i}^{2}(k)-2h_{i}(k)\big)\Big)
\end{aligned}	
\end{equation}
since $h_{i}(k)$ is a constant once we know $\vec{s}(k)$.
Moreover, from the definition of MW policies, in each slot, $P(A(k)=i)\in\left\{0,1\right\}$ and $\sum_{i=1}^{3}P(A(k)=i)=1$. Therefore, we have
\begin{align}
&\mathbb{E}[d_{i}(k)|\vec{s}(k)]=1_{\{A(k)=i\}}(1-\epsilon_{i})\label{equ: MW upper bound Ed}\\
&\mathbb{E}[t_{i}(k)|\vec{s}(k)]=1_{\{A(k)=3\}}1_{\{w_{i}(k)>0\}}(1-\epsilon_{i}).\label{equ: MW upper bound Et}
\end{align}
Substituting \eqref{equ: MW upper bound Ed} and \eqref{equ: MW upper bound Et} into \eqref{equ: MW-upper bound-tilde delta}, we obtain the desired results.
\end{proof}

\section{Proof of Theorem \ref{thm: MW-the upper bound}}\label{App: Theorem 5}

Before obtaining the upper bound of Max-Weight policies, we consider any randomized policy with probability $\mu_{1}$, $\mu_{2}$, and $\mu_{3}$. Now we calculate the probability of $Q_{2}^{(i)}=\varnothing$. This can be found as a corollary of the following lemma:

\begin{lemma}\label{lem: empty queue}
Consider a queue $Q$ of size $1$. Suppose that packets arrive with a Bernoulli random process of rate $\lambda$ and leave with another (independent) Bernoulli random process of rate $\mu$. Then the probability of empty queue is
$$P(Q=\varnothing)=\frac{\mu(1-\lambda)}{\mu+\lambda-\mu\lambda}.$$
\end{lemma}
\begin{proof}
Since the capacity of the system is $1$, then only two state exists, one is ``an empty queue'', and the other is ``a non-empty queue''. Denote ``an empty queue'' as state~$0$ and ``a non-empty queue'' as state~$1$, and consider a two-state Markov chain. Then, we find the transition probabilities
\begin{equation*}
\begin{aligned}
P_{0\rightarrow1}=&P\{\text{a arrival occurs}\}=\lambda\\
P_{1\rightarrow0}=&P\text{\{a departure occurs but no arrivals occur\}}\\
=&\mu(1-\lambda). 
\end{aligned}
\end{equation*}
Thus the transition probability matrix is
\begin{align*}
P=\left[\begin{matrix}
1-\lambda&\lambda\\
\mu(1-\lambda)&1-\mu+\mu\lambda	
\end{matrix}\right].
\end{align*}
Denote the stationary distribution as $\pi=(\pi_{0},\pi_{1})$, therefore
\begin{align*}
\pi=\pi P\Rightarrow \pi_{0}=\frac{\mu(1-\lambda)}{\mu+\lambda-\mu\lambda}.
\end{align*}
\end{proof}

\begin{corollary}\label{cor: P_empty}
The probability of empty queue in $Q_{2}^{(i)}$ is 
\begin{equation}\label{equ: empty queue}
\footnotesize
\begin{aligned}
P_{empty}^{i}:=\frac{\mu_{3}(1-\epsilon_{i})(1-\mu_{i}(\epsilon_{i}-\epsilon_{12}))}{\mu_{3}(1-\epsilon_{i})+\mu_{i}(\epsilon_{i}-\epsilon_{12})-\mu_{3}\mu_{i}(1-\epsilon_{i})(\epsilon_{i}-\epsilon_{12})}.	
\end{aligned}	
\end{equation}
\end{corollary}
\begin{proof}
Since $Q_{2}^{(i)}$ is the queue in which packets are erased by user $i$ but received user $\backslash{i}$, so the arrival rate is $$\lambda=\mu_{i}(\epsilon_{i}-\epsilon_{12}),$$ and the departure rate $$\mu=\mu_{3}(1-\epsilon_{i}).$$ From Lemma \ref{lem: empty queue},
\begin{equation*}
\footnotesize
\begin{aligned}
P(Q_{2}^{(i)}=\varnothing)=&\frac{\mu_{3}(1-\epsilon_{i})(1-\mu_{i}(\epsilon_{i}-\epsilon_{12}))}{1-(1-\mu_{3}(1-\epsilon_{i}))(1-\mu_{i}(\epsilon_{i}-\epsilon_{12}))}\\
=&\frac{\mu_{3}(1-\epsilon_{i})(1-\mu_{i}(\epsilon_{i}-\epsilon_{12}))}{\mu_{3}(1-\epsilon_{i})+\mu_{i}(\epsilon_{i}-\epsilon_{12})-\mu_{3}\mu_{i}(1-\epsilon_{i})(\epsilon_{i}-\epsilon_{12})}.
\end{aligned}
\end{equation*}
\end{proof}

To obtain the upper bound of MW policies, we manipulate the expression of the one-slot Lyapunov Drift. From the proof of Theorem~\ref{thm: MW-policy}, we have
\begin{equation*}
\footnotesize
\begin{aligned}
\Theta\big(\vec{h}(k)\big)=&\frac{1}{2}\sum_{i=1}^{2}\alpha_{i}\Big(-\mathbb{E}[d_{i}(k)|\vec{s}(k)]\big(h_{i}^{2}(k)+2h_{i}(k)\big)+2h_{i}(k)+1\\
+&\mathbb{E}[t_{i}(k)|\vec{s}(k)]\big(w_{i}^{2}(k)+2w_{i}(k)-h_{i}^{2}(k)-2h_{i}(k)\big)\Big).
\end{aligned}	
\end{equation*}
Consider any randomized policy with probability $(\mu_{1},\mu_{2},\mu_{3})$ in Section \ref{sec: stationary randomized policy}, and denote the probability of the non-empty queue in $Q_{2}^{(i)}$ as $P_{ne}^{i}=1-P_{empty}^{i}$. From the definition of Max-Weight policies, and substituting $\mu_{1}$, $\mu_{2}$, $\mu_{3}$ and $P_{ne}^{i}$ into $\Theta(\vec{h}(k))$, we find 
\begin{equation*}
\footnotesize
\begin{aligned}
\Theta(\vec{h}(k))\leq&\frac{1}{2}\sum_{i=1}^{2}\alpha_{i}\big(-\mu_{i}(1-\epsilon_{i})(h_{i}^{2}(k)+2h_{i}(k))+2h_{i}(k)+1\\
+&\mu_{3}P_{ne}^{i}(1-\epsilon_{i})(w_{i}^{2}(k)+2w_{i}(k)-h_{i}^{2}(k)-2h_{i}(k))\big).
\end{aligned}
\end{equation*}
From Remark~\ref{w<=h-1}, since $w_{k}(k)\leq h_{i}(k)-1$, then
\begin{equation*}
\footnotesize
\begin{aligned}
\Theta(\vec{h}(k))\leq\frac{1}{2}\sum_{i=1}^{2}\alpha_{i}\big(-\mu_{i}(1-\epsilon_{i})(h_{i}(k)-\Phi_{i})^{2}+\Psi_{i}\big)
\end{aligned}	
\end{equation*}
where
\begin{equation*}
\footnotesize
\begin{aligned}
&\Phi_{i}=\frac{1-\mu_{3}P_{ne}^{i}(1-\epsilon_{i})-\mu_{i}(1-\epsilon_{i})}{\mu_{i}(1-\epsilon_{i})}\\
&\Psi_{i}=1-\mu_{3}P_{ne}^{i}(1-\epsilon_{i})+\frac{\big(1-\mu_{i}(1-\epsilon_{i})-\mu_{3}P_{ne}^{i}(1-\epsilon_{i})\big)^{2}}{\mu_{i}(1-\epsilon_{i})}.
\end{aligned}
\end{equation*}
By Cauchy-Schwarz inequality,
\begin{equation*}
\footnotesize
\begin{aligned}
&\sum_{i=1}^{2}\alpha_{i}\mu_{i}(1-\epsilon_{i})(h_{i}(k)-\Phi_{i})^{2}(\sum_{i=1}^{2}\frac{\alpha_{i}}{\mu_{i}(1-\epsilon_{i})})\\
\geq&\bigg(\sum_{i=1}^{2}\alpha_{i}|h_{i}(k)-\Phi_{i}|\bigg)^{2}.
\end{aligned}
\end{equation*} 
Therefore
\begin{equation*}
\footnotesize
\begin{aligned}
\Theta(\vec{h}(k))\leq&-\frac{1}{2}(\sum_{i=1}^{2}\frac{\alpha_{i}}{\mu_{i}(1-\epsilon_{i})})^{-1}\big(\sum_{i=1}^{2}\alpha_{i}|h_{i}(k)-\Phi_{i}|\big)^{2}\\
+&\frac{1}{2}\sum_{i=1}^{2}\alpha_{i}\Psi_{i},
\end{aligned}
\end{equation*}
which implies
\begin{equation*}
\footnotesize	
\begin{aligned}
&\frac{1}{2}\big(\sum_{i=1}^{2}\alpha_{i}|h_{i}(k)-\Phi_{i}|\big)^{2}\\
\leq&-(\sum_{i=1}^{2}\frac{\alpha_{i}}{\mu_{i}(1-\epsilon_{i})})\Theta(\vec{h}(k))+\frac{1}{2}\sum_{i=1}^{2}\frac{\alpha_{i}}{\mu_{i}(1-\epsilon_{i})}\sum_{i=1}^{2}\alpha_{i}\Psi_{i}.
\end{aligned}
\end{equation*}
Summing over $k\in\left\{1,2,\cdots,T\right\}$, taking expectation and dividing by $T$ results in
\begin{equation*}
\footnotesize
\begin{aligned}
&\frac{1}{2T}\bigg(\sum_{k=1}^{T}\mathbb{E}\big(\sum_{i=1}^{2}\alpha_{i}|h_{i}(k)-\Phi_{i}|\big)^{2}\bigg)\\
\leq& -\big(\sum_{i=1}^{2}\frac{\alpha_{i}}{\mu_{i}(1-\epsilon_{i})}\big)\frac{1}{T}\sum_{k=1}^{T}\mathbb{E}[\Theta(\vec{h}(k))]+\frac{1}{2}(\sum_{i=1}^{2}\frac{\alpha_{i}}{\mu_{i}(1-\epsilon_{i})})\sum_{i=1}^{2}\alpha_{i}\Psi_{i}.
\end{aligned}
\end{equation*}
By Jensen's inequality, 
\begin{equation*}
\footnotesize
\begin{aligned}
&\frac{1}{2}\bigg(\frac{1}{T}\sum_{k=1}^{T}\mathbb{E}\big(\sum_{i=1}^{2}\alpha_{i}|h_{i}(k)-\Phi_{i}|\big)\bigg)^{2}\\
\leq& -(\sum_{i=1}^{2}\frac{\alpha_{i}}{\mu_{i}(1-\epsilon_{i})})\frac{1}{T}\sum_{k=1}^{T}\mathbb{E}[\Theta(\vec{h}(k))]+\frac{1}{2}\sum_{i=1}^{2}\frac{\alpha_{i}}{\mu_{i}(1-\epsilon_{i})}\sum_{i=1}^{2}\alpha_{i}\Psi_{i}.
\end{aligned}
\end{equation*}
Since
\begin{align*}
-\frac{1}{T}\sum_{k=1}^{T}\mathbb{E}[\Theta(\vec{h}(k))]\leq\frac{\mathbb{E}[L(\vec{h}(1))]}{T}
\end{align*}
and $\mathbb{E}[L(\vec{h}(1))]$ is a constant, then 
\begin{align*}
\lim_{T\rightarrow\infty}-\frac{1}{T}\sum_{k=1}^{T}\mathbb{E}[\Theta(\vec{h}(k))]=0.
\end{align*}
Thus,
\begin{equation*}
\footnotesize
\begin{aligned}
\mathbb{E}[J^{MW}]\leq \sqrt{\frac{1}{2}\sum_{i=1}^{2}\frac{\alpha_{i}}{\mu_{i}(1-\epsilon_{i})}\sum_{i=1}^{2}\alpha_{i}\Psi_{i}}+\frac{1}{2}\sum_{i=1}^{2}\alpha_{i}\Phi_{i}	
\end{aligned}	
\end{equation*}
where 
\begin{equation*}
\footnotesize
\begin{aligned}
&\Phi_{i}=\frac{1-\mu_{3}P_{ne}^{i}(1-\epsilon_{i})-\mu_{i}(1-\epsilon_{i})}{\mu_{i}(1-\epsilon_{i})}\\
&\Psi_{i}=1-\mu_{3}P_{ne}^{i}(1-\epsilon_{i})+\frac{\big(1-\mu_{i}(1-\epsilon_{i})-\mu_{3}P_{ne}^{i}(1-\epsilon_{i})\big)^{2}}{\mu_{i}(1-\epsilon_{i})}.
\end{aligned}
\end{equation*}

\bibliographystyle{IEEEtran}
\bibliography{reference}

\end{document}